\documentclass[a4paper]{article}
\usepackage[a4paper, total={5.5in, 9in}]{geometry}
\usepackage{authblk}
\usepackage{bm}
\usepackage{float}
\usepackage[T1]{fontenc}
\usepackage{changepage}
\usepackage[style=default]{caption}
\usepackage[noEnd=true]{algpseudocodex}
\usepackage{booktabs}
\usepackage[ruled]{algorithm}
\usepackage{graphicx}
\usepackage[shortlabels]{enumitem}
\usepackage{fontawesome}
\usepackage{multirow}
\usepackage{microtype}
\usepackage{wrapfig}
\usepackage{hyperref}
\hypersetup{colorlinks=true, citecolor=teal, linkcolor=teal, urlcolor=teal}
\usepackage[osf,sc]{mathpazo}
\usepackage[scaled=0.85]{helvet}
\usepackage{array}
\usepackage{tikzscale}
\usepackage[kw]{pseudo}
\pseudoset{
  compact,
  label=\scriptsize\sffamily\color{gray}\arabic*,
  ctfont=\sffamily,
  ct-left=\quad\ctfont{[},
  ct-right=\ctfont{]},
}

\setlist[description]{parsep=0.1px}

\usepackage[shortlabels]{enumitem}

\usepackage{amsmath,amsfonts,amssymb,amsthm}
\usepackage{cleveref}
\newtheorem{definition}{Definition}
\numberwithin{definition}{section}
\newtheorem{theorem}[definition]{Theorem}

\newtheorem{lemma}[definition]{Lemma}
\theoremstyle{remark}
\newtheorem{remark}{Remark}[section]
\newtheorem{example}{Example}[section]

\usepackage[
style=alphabetic,
backend=bibtex,
isbn=false,
alldates=year,
sortcites=true,
doi=true,
eprint=false
]{biblatex}
\bibliography{paper.bib}

\AtEveryBibitem{
    \clearfield{urlyear}
    \clearfield{urlmonth}
}

\renewbibmacro*{doi+url}{%
  \iftoggle{bbx:doi}
  {\printfield{doi}%
    \iffieldundef{doi}{}{}}
  {}%
  \newunit\newblock
  \iftoggle{bbx:eprint}
  {\usebibmacro{eprint}%
    \iffieldundef{eprint}{}{}}
  {}%
  \newunit\newblock
  \iftoggle{bbx:url}
  {\usebibmacro{url+urldate}%
    \iffieldundef{url}{}{}}
  {}}

\AtEveryBibitem{\clearfield{month}}
\AtEveryBibitem{\clearfield{day}}
\DeclareNameAlias{sortname}{first-last}



\newcommand{\set}[1]{\left\{#1\right\}}
\newcommand{\setdes}[2]{\left\{#1\;\middle|\;#2\right\}}

\def\xx{\mathbf{x}}
\def\yy{\mathbf{y}}

\def\tor{\mathbb{T}}
\def\Var{\mathbf{V}}

\def\NN{\mathbb{N}}
\def\RR{\mathbb{R}}
\def\CC{\mathbb{C}}
\def\ZZ{\mathbb{Z}}

\def\prja{\pi}
\def\prjb{p}
\def\vol{\operatorname{vol}}
\def\poly{\Delta}
\def\fc{\Phi}
\def\nc{\operatorname{N}}
\def\minpol{f_{\text{min}}}

\newcommand{\affspan}[1]{\langle #1 \rangle}

\newcommand{\mint}[1]{M_{\prjb}(#1)}
\newcommand{\mmo}{F_{\prjb}}
\newcommand{\conv}[1]{\operatorname{Conv}(#1)}

\newcommand{\type}{\operatorname{type}}
\newcommand{\abs}[1]{\left|#1\right|}
\newcommand{\norm}[1]{\left\lVert#1\right\rVert}
\newcommand{\twov}[2]{\begin{pmatrix} #1 \\ #2 \end{pmatrix}}

\DeclareMathOperator{\argmin}{argmin}
\DeclareMathOperator{\trop}{trop}

\definecolor{col0}{rgb}{1.0,0.0,0.0}
\definecolor{col1}{rgb}{0.0,1.0,0.0}
\definecolor{col2}{rgb}{0.0,0.0,1.0}
\definecolor{dark}{rgb}{0.4,0.4,0.4}

\title{On the Computation of Newton Polytopes of Eliminants}

\date{}

\author{Rafael Mohr \thanks{\href{mailto:rafael.mohr@inria.fr}{rafael.mohr@inria.fr}, Inria Saclay, Palaiseau, France}
  \and Yulia Mukhina \thanks{\href{mailto:yulia.mukhina@lix.polytechnique.fr}{yulia.mukhina@lix.polytechnique.fr}, LIX, CNRS, École Polytechnique, Palaiseau, France}
  }

\begin{document}

\maketitle

\begin{abstract}
  For systems of polynomial equations, we study the problem of
  computing the Newton polytope of their eliminants. As was shown by
  Esterov and Khovanskii, such Newton polytopes are {\em mixed fiber
    polytopes} of the Newton polytopes of the input equations. We use
  their results in combination with {\em mixed subdivisions} to design
  an algorithm computing these special polytopes.  We demonstrate the
  increase in practical performance of our algorithm compared to
  existing methods using tropical geometry and discuss the differences
  that lead to this increase in performance. We also demonstrate an
  application of our work to differential elimination.
\end{abstract}

\section{Introduction}
\label{sec:intro}

\paragraph{Problem Statement \& Motivation.}

Let $F:=\set{f_0,\dots,f_k}\in \CC[\xx^{\pm},\yy^{\pm}]$ be a sequence of
$k+1$ Laurent polynomials in two sets of variables
$\xx:=\set{x_1,\dots,x_{n-k}}$ and $\yy:=\set{y_1,\dots,y_{k}}$. If
$F$ defines a complete intersection $X:=\Var(F)$ in the
$n$-dimensional torus $\tor^n:=(\CC\setminus \set{0})^n$, then the projection
of $X$ to the $n-k$-dimensional torus $\tor_{\xx}^{n-k}$ corresponding
to the variables $\xx$ is a hypersurface $\Var(g)$,
$g\in \CC[\xx^{\pm}]$. Recall that the {\em Newton polytope} of a
polynomial is defined as the convex hull of its support.

The goal of this paper is to obtain a new algorithm which computes the
Newton polytope of $g$ from the Newton polytopes of the $f_i$.
Such an algorithm can be used, for example, to estimate the size of
$g$ and to check whether it is feasible to compute $g$ using symbolic
elimination techniques such as Gröbner bases (see e.g. Chapters 2 and
3 in \cite{cox2015a}).

Moreover, if there is a way to sample points on $X$, then knowing the
Newton polytope of $g$ can be used to interpolate $g$ by solving a
linear system, see e.g. \cite{rose2025} for details. This is in
particular the case for {\em implicitization problems}: in this case
$n-k=k+1$ (i.e. there is one more variable in $\xx$ than in $\yy$) and
the $f_i$ are of the form $f_i=x_i-h_i$ for certain
$h_i\in \CC[\yy^{\pm}]$. The problem of eliminating the variables in
$\yy$ is then equivalent to finding a defining equation $g$ for the image of
the map
\begin{align*}
  \varphi :\; &\tor^k\rightarrow \tor^{k+1}\\
  &y\mapsto (h_0(y),\dots,h_k(y)).
\end{align*}

\paragraph{Methodology \& Main Results.}

Our algorithm is fundamentally based on the work by
\cite{esterov2008}. In this work it is shown that, if the coefficients
of the $f_i$ are generic and upon an appropriate choice of $g$, the
Newton polytope of $g$ is the {\em mixed fiber polytope}
$\mmo(\poly_0,\dots,\poly_k)\subset \RR_{\xx}^{n-k}$ of the Newton polytopes
$\poly_i$ of the $f_i$, a notion that was independently also
introduced in \cite{mcmullen2004}. This construction is based on the
famous BKK-Theorem (\cite{bernshtein1975, kouchnirenko1976,
  khovanskii1995}, see also \cite{cox2005}) which estimates the number
of nonzero roots of a square polynomial system in terms of the {\em
  mixed volume} of its Newton polytopes. Mixed volumes relate to
volumes (as a certain multilinear, symmetric extension of volumes to
multiple arguments) as mixed fiber polytopes relate to {\em fiber
  polytopes}. The notion of fiber polytopes was introduced in
\cite{billera1992}. As the name suggests, fiber polytopes relate to
the fibers of a projection of polytopes. A useful way to think about
them is that the fiber polytope of a linear projection of polytopes
$\prjb : \poly\rightarrow \prjb(\poly)$ is the ``average fiber'' of this
projection. In our case, the projection with respect to which we will
compute {\em mixed} fiber polytopes is the projection
$\prjb:\RR^n\rightarrow \RR_{\yy}^k$ which maps support vectors of monomials in
$\CC[\xx^{\pm},\yy^{\pm}]$ to their entries corresponding to the variables
in $\yy$.

As was shown in the seminal work \cite{huber1995}, mixed volumes can
be computed from {\em mixed subdivisions} of collections of
polytopes. Roughly speaking, a subdivision of a collection of
polytopes is a certain partition of the Minkowski sum of these
polytopes into smaller polytopes (called {\em cells}). In
\cite{huber1995} it is shown that, given a {\em mixed} subdivision, the
mixed volume is a sum of the usual volumes of some of the cells of
this subdivision. Numerous algorithms exist to compute mixed
subdivisions, see e.g. \cite{verschelde1996, mizutani2007,
  malajovich2017, gao2000, dyer1998, chen2019, jensen2016}.

We design an analogous approach for mixed fiber polytopes: Theorem 12
in \cite{esterov2008} describes the vertices of the mixed fiber
polytope $\mmo(\poly_0,\dots,\poly_k)$ of the Newton polytopes
$\poly_0,\dots,\poly_k$ of $f_0,\dots,f_k$ as sums of mixed fiber
polytopes associated to certain $k$-dimensional faces of the Minkowski
sum $\poly$ of the $\poly_i$. For such a $k$-dimensional face the
associated mixed fiber polytope consists of a single point. The
relevant $k$-dimensional faces of $\poly$ induce a {\em coherent}
subdivision of the images $\prjb(\poly_i)$. Our main theoretical
result is then a formula to compute the desired vertex of
$\mmo(\poly_0,\dots,\poly_k)$ from a mixed subdivision refining this
given coherent subdivision (\Cref{thm:main}).

Based on \Cref{thm:main} we obtain a new algorithm to compute vertices
of mixed fiber polytopes (\Cref{alg:vert}). This algorithm is a {\em
  vertex oracle} for the mixed fiber polytope of interest, i.e. it
computes, given a generic covector $\gamma\in (\RR^{n-k})^{*}$, the vertex of
the mixed fiber polytope on which $\gamma$ is minimized. From such a vertex
oracle, all vertices can then be recovered using the algorithm given
in \cite{huggins2006}.

\paragraph{Related Works \& Context.}

One way to solve the problem we are interested in is based on
techniques using {\em sparse resultants} (see e.g. \cite{emiris2012,
  emiris2010, emiris2007, emiris2003, emiris2013} and
\cite{gelfand1994} for a comprehensive introduction to sparse
resultants). Namely, the Newton polytope of the desired eliminant $g$
is a projection of the Newton polytope $\poly_R$ of the sparse
$(\prjb(\poly_0),\dots,\prjb(\poly_k))$-resultant. This latter
polytope $\poly_R$ can be computed by enumerating its vertices using a
combinatorial procedure.
Note that the Newton polytope $\poly_R$ lives in a space of dimension
equal to the sum of the number of vertices of the $\prjb(\poly_i)$,
making the enumeration of all its vertices potentially burdensome. For
implicitization of curves or surfaces this is solved in
\cite{emiris2007,emiris2010}, where the enumeration procedure can be
refined to directly compute the needed projection of $\poly_R$.

We further mention the work \cite{hauenstein2014} which, compared to
our algorithm, includes also a numerical component. Here the Newton
polytope of a polynomial $g$ is computed when one can evaluate $g$ or
produce {\em witness sets} on $\Var(g)$, i.e. points on the
intersection of $\Var(g)$ with a linear subspace of suitable
dimension.

The strategy to compute Newton polytopes of eliminants that is closest
to our method uses {\em tropical elimination theory}
\cite{maclagan2015, sturmfels2008, sturmfels2007, rose2025,
  sturmfels2007a}. Section 4 in \cite{sturmfels2007} explicitly
demonstrates that mixed fiber polytopes can be computed using tropical
geometry. To highlight the differences with our algorithm let us
sketch partially how tropical elimination works. With generic input
$F:=\set{f_0,\dots,f_k}\in \CC[\xx^{\pm},\yy^{\pm}]$ and $X:=\Var(F)$ one
starts by computing the {\em tropicalization} $\trop(X)\subset \RR^n$. Here,
$\trop(X)$ consists of the support of a subfan of the {\em normal fan}
of the Minkowski sum of the Newton polytopes $\poly_i$ of the $f_i$
together with associated {\em multiplicities}. These multiplicities
are certain mixed volumes of faces of the $\poly_i$. The
tropicalization $\trop(X)$ is then projected to the space
$\RR^{n-k}_{\xx}$. This projection yields again a vertex oracle
for $\mmo(\poly_0,\dots,\poly_k)$.

While tropical elimination works for more general problems than
considered here (see e.g. Section 3 in \cite{rose2025}, i.e. also for
certain classes of varieties $X$ which are not complete
intersections), our algorithm, which does not use any tropical
geometry, offers two practical advantages in the considered setting:
First, the mixed volumes needed to compute $\trop(X)$ are computed
from mixed subdivisions in dimension $k+1$. Our algorithm necessitates
the computation of mixed subdivisions in dimension $k$, i.e. dimension
one less.  A more substantial improvement comes from the fact the we
avoid explicit computations with the Minkowski sum of the
$\poly_i$. To illustrate the resulting improvement, consider the
following example:
\begin{example}
  \label{ex:intro}
  In $\CC[x_0,x_1,x_2,x_3,y_1,y_2,y_3]$ we let
  \begin{align*}
    &f_0:=x_0-(y_1y_2 + y_3 + 1),f_1:=x_1-(y_1y_3+y_2+1),\\
    &f_2:=x_2-(y_1y_3 + y_1 + 1),f_3:=x_3-(y_1^3+y_2^5+y_3^7).
  \end{align*}
  This example appears on page 113 of \cite{sturmfels2008}.  Let $X$
  be the algebraic set defined by the $f_i$ and let $\poly_i$ be the
  Newton polytope of $f_i$ for each $i=0,\dots,3$. The Minkowski sum
  $\poly:=\poly_0+\dots +\poly_4$ has 651 faces of dimension $4$. The
  computation of $\trop(X)$ requires the computation of one mixed
  volume for each of these faces. Our algorithm extracts the mixed
  fiber polytope $\mmo(\poly_0,\dots,\poly_3)$ with just 38
  computations of mixed subdivisions.
\end{example}
Indeed, note that the Minkowski sum of $\poly_0,\dots,\poly_k$ can
have a number of faces in the worst case equal to the product of the
number of faces of the $\poly_i$, see e.g. \cite{grunbaum2003}.

We highlight these two practical advantages experimentally with a set
of examples, see \Cref{sec:exam}. These experiments are performed with
an implementation of our algorithm in the programming language
\texttt{julia} \cite{bezanson2017}, comparing it to
\texttt{TropicalImplicitization.jl} which is written in \texttt{julia}
as well. This latter package implements the computation of mixed fiber
polytopes based on tropical geometry for implicitization problems, as
in \cite{rose2025, sturmfels2008, sturmfels2007}. On the chosen
examples our software outperforms tropical implicitization by several
orders of magnitude.

\paragraph{Organization of the Paper.}

In \Cref{sec:found}, we introduce the necessary theory on mixed fiber
polytopes, based on the results in \cite{esterov2008}, as well as the
appropriate notion of subdivisions suitable for our work. We attempt
to be as self-contained as possible with respect to the necessary
polyhedral geometry and refer e.g. to \cite{ziegler1995} for a
comprehensive introduction to the subject. In \Cref{sec:main} we then
prove our main theoretical result and give the resulting algorithm
(\Cref{alg:vert}).  \Cref{sec:exam} gives more details on our
implementation of \Cref{alg:vert} as well as several examples on which
we ran our software. \Cref{sec:diffelim} discusses an application of
our method to differential algebra and gives further relevant
experimental data.

\paragraph*{Acknowledgements}
\label{par:ack}

The authors wish to thank Gleb Pogudin and Pierre Lairez for helpful
discussions and suggestions. RM was supported by the ERC project
``10000 DIGITS'' (ID 101040794). YM was supported by the French
ANR-22-CE48-0008 OCCAM and ANR-22-CE48-0016 NODE projects.

\section{Theoretical Foundations}
\label{sec:found}

As in the introduction, fix a sequence of $k+1$ Laurent polynomials
$F:=\set{f_0,\dots,f_k}\in \CC[\xx^{\pm},\yy^{\pm}]$ in two sets of
variables $\xx:=\set{x_1,\dots,x_{n-k}}$ and
$\yy:=\set{y_1,\dots,y_{k}}$. Assume $F$ defines a complete
intersection $X$ and that $f_i$ has Newton polytope
$\poly_i$. The Zariski closure of the projection of $X$ to
$\tor_{\xx}^{n-k}$ is a hypersurface $\Var(g)$, $g\in \CC[\xx^{\pm}]$. We
want to compute the Newton polytope of $g$ from the Newton polytopes
of the $f_i$.

The main result of \cite{esterov2008} is that the Newton polytope of
an appropriate choice of $g$ is contained in the \textit{mixed fiber
  polytope} of $\poly_0,\dots,\poly_k$. Given $\poly_0,\dots,\poly_k$,
this is the object we want to compute.

\subsection{Mixed Fiber Polytopes}
\label{sec:mfp}

For the remainder of this work, fix an integer $n\in \NN$ and let
$\prja : \RR^n\rightarrow \RR^{n-k}$ be the projection on the first
$n-k$ coordinates. Let $\prjb: \RR^n\rightarrow \RR^k$ be the projection on the
last $k$ coordinates.

For a set $S\subset \RR^n$, we denote by
$\affspan{S}$ the \textit{affine span} of $S$, i.e. the smallest
affine subspace of $\RR^n$ containing $S$. If $S$ is finite, then
we denote by $\conv{S}\subset \RR^n$
the \textit{convex hull} of $S$, i.e.
\[\conv{S} := \setdes{\sum_{s\in S} \alpha_ss}{\alpha_s\in [0,1], \sum_{s\in S}\alpha_s = 1}.\]
For us, a {\em polytope} is a subset of $\RR^n$ which can be written
as a convex hull of finitely many points. The {\em dimension} of a
polytope is the dimension of its affine span.

\begin{definition}[Minkowski Integral, Fiber Polytope]
  \label{def:fibpol}
  For a polytope $\poly\subset \RR^n$ define
  \begin{enumerate}
  \item the {\em Minkowski integral} as
    \[\mint{\poly} :=
      \setdes{\int_{p(\poly)}s(x)dx}{\substack{s\;:\;p(\poly)\rightarrow \poly\\
          \text{is a continous section of }\prjb}}\subset \RR^n.\]
  \item and the {\em fiber polytope} of $\poly$ as
    \[\mmo(\poly) := \pi\left(\mint{\poly}\right).\]
  \end{enumerate}
\end{definition}

Fiber polytopes were introduced in \cite{billera1992}. Our definition of
fiber polytopes follows \cite{esterov2008} and differs slightly from
the original one given in \cite{billera1992}: Here, the fiber polytope
is defined as
\[\frac{1}{\vol(\prjb(\poly))}\cdot \mint{\poly}.\]
In \cite{billera1992} it is in particular shown that the Minkowski
integral and fiber polytope of a polytope $\poly\subset \RR^n$ are polytopes
themselves.

\begin{example}
  \label{ex:fibpol}
  In \Cref{fig:fibpol} we attempted to sketch a fiber polytope
  $\mmo(\poly)$ in dimension $n=2$. Importantly, in this figure, the
  lower vertex $v_0$ of $\mmo(\poly)$ is the image under $\prja$ of
  the integral of the section corresponding to the lower three edges
  of $\poly$ (i.e. the ones on which the covector $1\in (\RR^1)^{*}$ is
  fiberwise minimized), while the upper vertex $v_1$ of $\mmo(\poly)$
  is the image under $\prja$ of the integral of the section
  corresponding to the upper three edges of $\poly$ (i.e. the ones on
  which the covector $-1\in (\RR^1)^{*}$ is fiberwise minimized). This
  observation was central in \cite{billera1992} and will be important
  for us too. We refer to Chapter 9 in \cite{ziegler1995} for further
  exposition on fiber polytopes and in particular to Theorem 9.6
  therein.
\end{example}

The basis of the definition of {\em mixed} fiber polytopes is the following
lemma.

\begin{lemma}[Theorem 10 in \cite{esterov2008}]
  \label{lem:prob}
  Fix variables $\lambda_0,\dots,\lambda_k$ and let
  $\poly_0,\dots,\poly_k\subset \RR^n$ be a collection of polytopes. Then
  $\mmo(\lambda_0\poly_0+\dots+\lambda_k\poly_k)$ is a homogeneous polynomial of
  degree $k+1$ with coefficients polytopes contained in $\RR^{n-k}$.
\end{lemma}

\begin{wrapfigure}{l}{0.6\textwidth}
  \centering
  \begin{tikzpicture}[thick,fill opacity=0.9,scale=0.6,every node/.style={scale=0.8}]
    
    % polytope
    \draw[line width=2pt][-, black] (0,1) -- (1,0) -- (5,0) -- (6,1) -- (6,2) -- (5,3) -- (1,3) -- (0,2) -- cycle; 
    \draw (5.8,2.8) node[]{\Large $\poly$};

    % projection of polytope
    \draw[line width=2pt][-, black] (0,-2) -- (6,-2);
    \draw (6,-2) node[above]{\Large $\prjb(\poly)$};

    % OX
    \draw[->, black] (-1,-2) -- (7.5,-2);
    \draw (7.5,-2) node[right]{\Large $\RR^{k}$};
    
    % OY
    \draw[->, black] (-1,-2) -- (-1,5);
    \draw (-1,5) node[above]{\Large $\RR^{n-k}$};

    % Fiber polytope
    \draw[line width=2pt][-, col0] (-1, 0.3) -- (-1, 2.7);
    \draw (-1,1.5) node[left]{\Large $\mmo(\poly)$};

    \node[circle, draw=black, fill=black, inner sep=0pt,minimum size=0.4em] at (-1,0.3) { };
    \draw (-1,0.3) node[left]{\Large $v_0$};
    \draw (-1,2.7) node[left]{\Large $v_{1}$};
    \node[circle, draw=black, fill=black, inner sep=0pt,minimum size=0.4em] at (-1,2.7) { };

    % p
    \draw[-latex, black] (7,2.5) -- (7,0.5);
    \draw (7,1.5) node[right]{\Large $\prjb$};

    % pi
    \draw[-latex, black] (4,4) -- (2,4);
    \draw (3,4) node[above]{\Large $\prja$};

  \end{tikzpicture}
  \caption{A Fiber Polytope $\mmo(\poly)$.}
  \label{fig:fibpol}
\end{wrapfigure}

The summation in the polynomial
$\mmo(\lambda_0\poly_0+\dots+\lambda_k\poly_k)$ introduced in \Cref{lem:prob} is
understood as Minkowski summation.

\begin{remark}
  In order for \Cref{lem:prob} to be true, we have to consider $0$ as a
  homogeneous polynomial of degree $k+1$ as well. If
  \[\dim \prjb(\affspan{\poly_0+ \dots + \poly_k}) < k,\]
  then
  \[\mmo(\lambda_0\poly_0+\dots+\lambda_k\poly_k) = \set{0}.\]
  Indeed, for all $\lambda_0,\dots,\lambda_k>0$, since
  $\prjb(\poly_0+\dots + \poly_k)$ is a set of Lebesgue measure zero,
  the Minkowski integral of $\lambda_0\poly_0+\dots+\lambda_k\poly_k$ is the
  singleton $\set{0}$.
\end{remark}

\Cref{lem:prob} yields the definition of our objects of interest:

\begin{definition}[Mixed Fiber Polytope]
  \label{def:mfp}
  For a collection of polytopes $\poly_0,\dots,\poly_k\subset \RR^n$, the
  \textit{mixed fiber polytope} $\mmo(\poly_0,\dots,\poly_k)$ is the
  coefficient of the monomial $\lambda_0\dots\lambda_k$ in the polynomial
  $\mmo(\lambda_0\poly_0+\dots+\lambda_k\poly_k)$.
\end{definition}

\begin{remark}
  Similarly to the well known notion of mixed volumes, the mixed fiber
  polytope is symmetric and multilinear in its arguments. This follows
  from the fact that it is the coefficient of
  $\lambda_0\dots \lambda_k$ in the polynomial
  $\mmo(\lambda_0\poly_0+\dots+\lambda_k\poly_k)$ \cite[Theorem 8
  in][]{esterov2008}. Note that this also implies
  \[(k+1)!\mmo(\poly) = \mmo(\poly,\dots,\poly).\]
\end{remark}

Our main interest in mixed fiber polytopes lies in the following
theorem:

\begin{theorem}[Definition 5, Theorem 2 and Theorem 7 in \cite{esterov2008}]
  \label{thm:mfpelim}
  Let \[F:=\set{f_0,\dots,f_k}\subset \CC[\xx^{\pm},\yy^{\pm}]\] be Laurent
  polynomials with Newton polytopes $\poly_0,\dots,\poly_k$. Let
  $X := \Var(F)\subset \tor^n$ and let
  $\prja : \tor^n\rightarrow \tor^{n-k}_{\xx}$ be the projection to the torus
  corresponding to the variables in $\xx$. Suppose that $F$ defines a
  complete intersection. Then there exists $g\in \CC[\yy^{\pm}]$ such that
  \[\overline{\prja(X)} = \Var(g),\]
  and such that the Newton polytope of $g$ is contained in
  $\mmo(\poly_0,\dots,\poly_k)$. If the coefficients of $F$ are
  sufficiently generic, then the Newton polytope of $g$ is equal to
  $\mmo(\poly_0,\dots,\poly_k)$.
\end{theorem}

\begin{remark}
  The polynomial $g$ in \Cref{thm:mfpelim} may be exactly
  characterized as the {\em composite polynomial} of $F$, see
  Definition 2 in \cite{esterov2008}.
\end{remark}

\subsection{Vertices of Mixed Fiber Polytopes}
\label{sec:mfpvert}

Let $\poly_0,\dots,\poly_k\subset \RR^n$ be polytopes. We will compute the
mixed fiber polytope $\mmo(\poly_0,\dots,\poly_k)$ by providing a
\textit{vertex oracle}, i.e. a function that returns for a generic
covector $\gamma \in (\RR^{n-k})^{*}$ the unique point of
$\mmo(\poly_0,\dots,\poly_k)$ on which $\gamma$ is minimized. Note that any
polytope is the convex hull of its vertices. Let us first introduce
some associated notation and recall some standard definitions:

For a compact subset $S\subset \RR^n$ and covector
$\gamma\in (\RR^n)^{*}$ we denote by $S^{\gamma}$ the set of points in
$S$ where $\gamma$ is minimized.

\begin{definition}[Face, Vertex, Normal Cone]
 \label{def:face}
 A \textit{face} of a polytope $\poly\subset \RR^n$ is a set of the form
 $\poly^{\gamma}$ for some $\gamma\in (\RR^n)^{*}$. Every face of a polytope is 
 a polytope itself.

 The \textit{vertices} of $\poly$ are the $0$-dimensional
 faces, they consist of singletons.

 The \textit{normal cone}
 $\nc(\fc,\poly)$ of a face $\fc$ of $\poly$ is defined as
 \[\nc(\fc,\poly):=\setdes{\gamma\in (\RR^n)^{*}}{\fc \subseteq \poly^{\gamma}},\]
\end{definition}

Any normal cone is a {\em convex polyhedral cone}, i.e. it can be
written as the set of all linear combinations with non-negative
coefficients of a finite set of vectors. Recall also that, for a
polytope $\poly\subset \RR^n$ and a face $\fc$ of $\poly$, we have
$\dim (\nc(\fc,\poly)) = n - \dim(\fc)$. Here, the dimension of a
convex polyhedral cone is again defined to be the dimension of its
affine span. This implies in particular that $\poly^{\gamma}$ is a vertex
for all $\gamma\in (\RR^n)^{*}$ lying outside a union of finitely many
subspaces of $(\RR^n)^{*}$.

Recall further that, for a collection of polytopes
$\poly_0,\dots,\poly_k\subset \RR^n$, the faces of the Minkowski sum
$\poly:=\poly_0+\dots+\poly_k$ are sums of faces of the $\poly_i$: For
a covector $\gamma\in (\RR^n)^{*}$ we have
\[\poly^{\gamma} = \poly_0^{\gamma}+ \dots + \poly_k^{\gamma}.\]
Thus, by a face of a tuple of polytopes
$\poly_{\bullet}:=(\poly_0,\dots,\poly_k)$ we mean a tuple of the form
$\fc_{\bullet}:=(\poly_0^{\gamma},\dots,\poly_k^{\gamma})$ for a covector
$\gamma\in (\RR^n)^{*}$. By the dimension of $\fc_{\bullet}$ we mean the dimension of
the corresponding face of $\poly$ and by the normal cone
$N(\fc_{\bullet},\poly_{\bullet})$ the normal cone of the corresponding face of
$\poly$. For some integer $d\in \NN$ we denote by
$\mathcal{F}_d(\poly_{\bullet})$ the set of $d$-dimensional faces of $\poly_{\bullet}$.

The basis of our approach is an appropriate mixed version of the
phenomenon that we illustrated in \Cref{ex:fibpol} which reads as
follows:

\begin{theorem}
  \label{thm:mfpvert}
  Let $\poly_{\bullet}:=(\poly_0,\dots,\poly_k)$ be a tuple of polytopes in
  $\RR^n$ and let $\poly:=\poly_0+\dots+\poly_k$.  For a covector
  $\gamma\in (\RR^{n-k})^{*}$ let
  \[\mathcal{F}_{k,\gamma} := \setdes{\fc_{\bullet}\in \mathcal{F}_k(\poly_{\bullet})}{\forall x\in \sum_{\fc_i\in \fc_{\bullet}}\fc_i : x = \argmin_{\gamma}\prjb^{-1}(\prjb(x))},\]
  this is the set of $k$-dimensional faces of $\poly_{\bullet}$ on which
  $\gamma$ is fiberwise minimized.

  If $\gamma$ is sufficiently generic, i.e. lies outside a union of
  finitely many linear subspaces of $(\RR^{n-k})^{*}$, then
  $\mmo(\poly_{\bullet})^{\gamma}$ is a vertex of
  $\mmo(\poly_{\bullet})$ which is given by
  \[\mmo(\poly_{\bullet})^{\gamma} = \sum_{\fc_{\bullet}\in \mathcal{F}_{k,\gamma}}
    \mmo(\fc_{\bullet}).\]
\end{theorem}
\begin{proof}
  Note that
  \[\mathcal{F}_{k,\gamma} =\setdes{\fc\in \mathcal{F}_k}{\gamma\in\prja(\nc(\fc,\poly))}.\]
  Now the stated claim is just Theorem 12 in \cite{esterov2008}.
\end{proof}

In the setting of \Cref{thm:mfpvert}, $\prjb$ induces a bijection from
$\sum_i\fc_i$ to $\prjb\left( \sum_i\fc_i\right)$ for any
$\fc_{\bullet}\in \mathcal{F}_{k,\gamma}$.  Hence
$\mmo(\fc_{\bullet})$ consists of a single point since the Minkowski
integral of $\sum_i\fc_i$ consists of a single point.

\subsection{Subdivisions of Point Configurations}
\label{sec:msubdiv}

A {\em subdivision} of a polytope $\poly$ is, roughly speaking, a set
of polytopes of the same dimension as $\poly$ which cover $\poly$ and
only intersect in faces of themselves.

If $\poly_0,\dots,\poly_k\subset \RR^n$ are polytopes and
$\poly:=\poly_0+\dots+\poly_k$, then we will apply \Cref{thm:mfpvert}
by computing {\em mixed subdivisions} of the polytope
$\prjb(\poly)$. Mixed subdivisions were introduced in \cite{huber1995}
to compute mixed volumes of polytopes. In the context of
\Cref{thm:mfpvert}, the images under $\prjb$ of the $k$-dimensional
faces of $\poly_{\bullet}$ on which a given covector
$\gamma\in (\RR^{n-k})^{*}$ is fiberwise minimized introduce another
subdivision of $\prjb(\poly)$, a {\em coherent} subdivision.

Let us now introduce the appropriate notion of subdivisions.  In
practice, the polytopes $\poly_0,\dots,\poly_k$ for which we want to
compute the mixed fiber polytope are given to us via supports of
polynomials of which they are the convex hull.  It is now more
convenient to work directly with the underlying points in the
supports.

\begin{definition}[Point Configuration, Cell]
  \label{def:ptconf}
  A {\em point configuration} in $\ZZ^n$ is a tuple of finite
  subsets $\mathcal{A}:=(\mathcal{A}_0,\dots,\mathcal{A}_k)$ of
  $\ZZ^{n}$. A {\em cell} of $\mathcal{A}$ is a tuple
  $C:=(C_0,\dots,C_k)$ of sets $C_i\subset \mathcal{A}_i$. Denote
  \begin{align*}
    &\type(C) := (\dim(\conv{C_0}),\dots,\dim(\conv{C_k})),\\
    &\conv{C} := \conv{C_0 + \dots + C_k}.
  \end{align*}
  A \textit{face} of a cell $C$ is a cell of the form
  $C^{\gamma}:=(C_0^{\gamma},\dots,C_k^{\gamma})$ for some $\gamma\in (\RR^n)^{*}$.
\end{definition}

\begin{definition}[Subdivision]
  \label{def:subdiv}
  Let $\mathcal{A}:=(\mathcal{A}_0,\dots,\mathcal{A}_k)$ be a point configuration in
  $\ZZ^{n}$. Let
  $d := \dim(\mathcal{A}) := \dim \affspan{\mathcal{A}_0+\dots+ \mathcal{A}_k}$. A {\em subdivision} $S$
  of $\mathcal{A}$ is a finite collection of cells of $\mathcal{A}$ s.t.
  \begin{enumerate}
  \item $\dim(\conv{C}) = d$ for all $C\in S$.
  \item $\conv{C_1}\cap \conv{C_2}$ is a face of both $C_1$ and $C_2$ for
    all $C_1,C_2\in S$.
  \item $\bigcup_{C\in S} \conv{C} = \conv{\mathcal{A}}$.
  \end{enumerate}
\end{definition}

Let $S$ be a subdivision of a point configuration $\mathcal{A}$. Each cell
$C\in S$ defines a polytope
$\conv{C}\subset \poly:=\conv{\mathcal{A}}$ of the same dimension as
$\poly$.  The polytopes $\conv{C}$, $C\in S$, divide $\poly$ into
smaller polytopes intersecting only in faces of themselves.

Let $\mathcal{A}:=(\mathcal{A}_0,\dots,\mathcal{A}_k)$ be a point configuration in
$\ZZ^n$ and let $\poly_i:=\conv{\mathcal{A}_i}$. Through \Cref{thm:mfpvert},
each vertex of the mixed fiber polytope $\mmo(\poly_0,\dots,\poly_k)$
will be determined by a {\em coherent} subdivision of the
point configuration
$\prjb(\mathcal{A}):=(\prjb(\mathcal{A}_0),\dots,\prjb(\mathcal{A}_k))$ which we now define.

\begin{definition}[Weight Vector, Lifted Configuration, Lower Facet]
  Let $\mathcal{A} = (\mathcal{A}_0,\dots,\mathcal{A}_k) \subset \ZZ^n$ be a point configuration.
  \begin{enumerate}
  \item A {\em weight vector for $\mathcal{A}$} is a vector
    \[w\in \RR^{\abs{\mathcal{A}_0}}\times \dots \times \RR^{\abs{\mathcal{A}_k}}\] indexed by the
    points in each $\mathcal{A}_i$.
  \item For a weight vector $w$ we define the {\em lifted configuration}
    $\mathcal{A}_w:=(\mathcal{A}_{0,w},\dots, \mathcal{A}_{k,w})$ of $\mathcal{A}$ as consisting of the sets
    \[\mathcal{A}_{i,w}:=\setdes{\twov{a}{w_a}}{a\in \mathcal{A}_i}\subset \RR^{n+1}\]
    for each $i=0,\dots,k$. Here, $w_a$ is the entry of $w$ corresponding
    to the point $a\in \mathcal{A}_i$.
  \item A {\em lower facet} $\fc$ of $\mathcal{A}_w$ is a face
    of dimension $\dim(\mathcal{A})$ of $\mathcal{A}_w$ such that any
    $\gamma\in (\RR^{n+1})^{*}$ with $\mathcal{A}_w^{\gamma} = \fc$ has positive last
    coordinate.
  \end{enumerate}
\end{definition}

\begin{definition}[Coherent Subdivision]
  \label{def:cohsub}
  Let $\mathcal{A}$ be a point configuration in $\ZZ^n$ and $w$ be a weight
  vector for $\mathcal{A}$. The {\em coherent subdivision} $S_w$ of
  $\mathcal{A}$ is the subdivision of $\mathcal{A}$ consisting of the images of all lower
  facets of the lifted configuration $\mathcal{A}_w$ under the projection
  forgetting the last coordinate.
\end{definition}

Coherent subdivisions are indeed subdivisions in the sense of
\Cref{def:ptconf}, see Lemma and Definition 2.6 in \cite{huber1995}.

 {\em Mixed subdivisions} of a point configuration
$\mathcal{A}$ in $\ZZ^n$ have been defined in \cite{huber1995} to compute mixed
volumes.

\begin{definition}[(Fine) Mixed Subdivision]
  \label{def:mixdiv}
  Let $S$ be a subdivision of a point configuration $\mathcal{A}$ of $\ZZ^n$.
  \begin{enumerate}
  \item The subdivision $S$ is called {\em mixed} if for all
    $C := (C_0,\dots,C_k)\in S$ we have
    \[\sum_j \dim(\conv{C_{j}}) = \dim(\mathcal{A}).\]
  \item The subdivision $S$ is called {\em fine mixed} if in addition,
    for all $C\in S$ we have
    \[\sum_j( \abs{C_j} - 1) = \dim(\mathcal{A}).\]
  \end{enumerate}
  Here, $\abs{C_j}$ denotes the cardinality of $C_j$.
\end{definition}

\begin{wrapfigure}{r}{0.6\textwidth}
  \centering
  \begin{tikzpicture}[thick,fill opacity=0.9,scale=0.6,every node/.style={scale=0.8}]
    \draw[line width=2pt][-, col0] (0,0) -- (5,0);
    \draw[line width=2pt][-, col2] (5,0) -- (8,0);

    \draw[dashed] (0,2) -- (0,0);
    \draw[dashed] (5,2.5) -- (5,0);
    \draw[dashed] (8,6) -- (8,0);

    \fill[gray,opacity=0.2] (0,2) -- (8,6) -- (5,2.5) -- cycle;

    \node[circle, draw=black, fill=black, inner sep=0pt,minimum size=0.4em] at (0,0) { };
    \draw (0,0) node[below]{$0+0$};
    \node[circle, draw=black, fill=black, inner sep=0pt,minimum size=0.4em] at (3,0) { };
    \draw (3,0) node[below]{$a+0$};
    \node[circle, draw=black, fill=black, inner sep=0pt,minimum size=0.4em] at (5,0) { };
    \draw (5,0) node[below]{$0+b$};
    \node[circle, draw=black, fill=black, inner sep=0pt,minimum size=0.4em] at (8,0) { };
    \draw (8,0) node[below]{$a+b$};
    \draw (4,0.7) node[below]{$\conv{\mathcal{A}}$};
    \draw (4,3) node[]{$\conv{\mathcal{A}_w}$};

    \draw[line width=1pt][-, black] (0,2) -- (8,6);
    \draw[line width=2pt][-, col0] (0,2) -- (5,2.5);
    \draw[line width=2pt][-, col2] (5,2.5) -- (8,6);

    \node[circle, draw=black, fill=black, inner sep=0pt,minimum size=0.4em] at (0,2) { };
    \node[circle, draw=black, fill=black, inner sep=0pt,minimum size=0.4em] at (3,3.5) { };
    \node[circle, draw=black, fill=black, inner sep=0pt,minimum size=0.4em] at (5,2.5) { };
    \node[circle, draw=black, fill=black, inner sep=0pt,minimum size=0.4em] at (8,6) { };
    
  \end{tikzpicture}
  \caption{A Coherent Mixed Subdivision}
  \label{fig:mixed}
\end{wrapfigure}
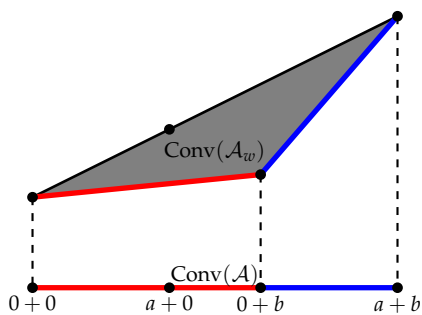

\begin{example}
  \label{ex:mixsub}
  For $a,b\in \ZZ$ with $a< b$ consider the point configuration $\mathcal{A}$ in $\ZZ$
  given by $\mathcal{A}_0:=\set{0,a}$ and $\mathcal{A}_1:=\set{0,b}$. The polytope
  $\conv{\mathcal{A}}$ is sketched as the line segment in \Cref{fig:mixed}. If we choose
  a suitable weight vector $w\in \RR^4$ with $w_a>w_b$ then the associated coherent
  subdivision has cells
  \begin{align*}
    &(\set{0},\set{0,b})\text{ and }\\
    &(\set{0,a},\set{b}).
  \end{align*}
  These correspond to the lower one-dimensional faces of the polytope
  $\conv{\mathcal{A}_w}$ as sketched in \Cref{fig:mixed} as well. The
  subdivision $S_w$ is fine mixed.
\end{example}

What we will want to do is to {\em refine} a given
coherent subdivision of a point configuration $\mathcal{A}$ in $\ZZ^n$ into a
fine mixed subdivision.

\begin{definition}[Refinement]
  Let $\mathcal{A}$ be a point configuration in $\ZZ^n$ and let
  $S_1,S_2$ be two subdivisions of $\mathcal{A}$. The subdivision
  $S_1$ is said to {\em refine} $S_2$ if each cell of $S_1$ is a
  subset of a cell of $S_2$.
\end{definition}

This is made possible by the following lemma:

\begin{lemma}
  \label{lem:refine}
  Let $\mathcal{A}$ be a point configuration in $\ZZ^n$ and suppose
  $\dim(\mathcal{A}) = n$. Let $S_w$ be a coherent subdivision of
  $\mathcal{A}$ associated to a weight vector $w$. For any sufficiently small
  $\epsilon>0$ and any sufficiently generic weight vector $w'$ such that
  $\norm{w-w'}< \epsilon$, the coherent subdivision $S_{w'}$ is fine mixed
  and refines $S_w$.
\end{lemma}

In \Cref{lem:refine}, ``sufficiently generic'' means ``lying outside a
union of finitely many proper linear subspaces''.

\begin{proof}

  The coherent subdivisions of $\mathcal{A}$ are in a refinement preserving
  bijection with the {\em regular subdivisions} \cite[Definition 5.3
  in][]{ziegler1995} of the  {\em Cayley embedding}
  $\mathcal{C}(\mathcal{A}_0,\dots,\mathcal{A}_k)$ \cite[Section 2.4, Theorem 3.1
  in][]{huber2000}. This bijection maps fine mixed subdivisions to
  regular triangulations, i.e. subdivisions where each cell is a
  simplex.

  We therefore have to prove only the following statement: Fix a
  finite set $\mathcal{B}\subset \ZZ^m$ s.t. the polytope
  $\poly = \conv{\mathcal{B}}$ is full-dimensional. Let $S_w$ be a regular
  subdivision of $\poly$ given by a lifting function $w$ on
  $\mathcal{B}$.  If $w'$ is another lifting function on $\mathcal{B}$ and
  $\norm{w-w'}<\epsilon$ for sufficiently small $\epsilon$ then the subdivision
  $S_{w'}$ of $\poly$ induced by $w'$ is a regular triangulation
  refining $S_w$.

  The covector $w$ lies in the normal cone of a face $\fc$ of the {\em
    secondary polytope} $\Sigma(\poly)$ \cite[Definition 1.6, Theorem 1.7
  in Chapter 7 of][]{gelfand1994}. If $\epsilon$ is sufficiently small, then
  $w'$ will lie in the interior of the normal cone of a vertex of
  $\Sigma(\poly)$ contained in $\fc$. This implies that $S_{w'}$ is a
  regular triangulation refining $S_w$ \cite[Theorem 2.4 in Chapter 7
  of][]{gelfand1994}.

  % For the fact that $S_{w'}$ is fine mixed for a sufficiently generic
  % weight vector $w'$, see the discussion below Example 2.7 in
  % \cite{huber1995}. Let $u:\RR^{n+1}\rightarrow \RR^n$ be the projection
  % forgetting the last coordinate. A face
  % $\fc':=\mathcal{A}_{w'}^{\gamma}$ of the lifted configuration
  % $\mathcal{A}_{w'}$ is a lower facet of $\mathcal{A}_{w'}$ iff it is of dimension
  % $\dim(\mathcal{A})$ and iff there is no choice of
  % $(x_0,\dots,x_k)\in \mathcal{A}_{0,w}\times \dots \times \mathcal{A}_{k,w}$ such that
  % $x:=x_0+\dots+x_k$ satisfies the following statements:
  % \begin{itemize}
  % \item $x\in \conv{\mathcal{A}_{w'}}\setminus \conv{\fc}$,
  % \item $x$ has minimal last coordinate in the fiber of $u$ over
  %   $u(x)$ in $\conv{\mathcal{A}_{w'}}$ and
  % \item $u(x)\in u(\conv{\fc})$,
  % \end{itemize}
  % see also
  % Definition 5.3 in \cite{ziegler1995}.

  % These two properties of $\fc'$ are preserved for the cell $\fc$ of
  % $\mathcal{A}_w$ corresponding to $\fc'$ if $w'$ is sufficiently generic and
  % $\epsilon$ is sufficiently small. Therefore $\fc$ must be contained in a
  % lower facet of $\mathcal{A}_w$, proving the statement.
\end{proof}

\section{Computing Mixed Fiber Polytopes}
\label{sec:main}

Analogously to mixed fiber polytopes, the mixed volume of a collection
of polytopes $\poly_1,\dots,\poly_r\subset \RR^n$ is the coefficient of
$\lambda_1\dots \lambda_r$ in the polynomial
$\vol(\lambda_1\poly_1+\dots+\lambda_r\poly_r)$ in variables
$\lambda_1,\dots,\lambda_r$.  A mixed subdivision of an underlying point
configuration of $\poly_1,\dots,\poly_r$ now allows to explicitly
compute this mixed volume as the sum of volumes of certain cells in
this subdivision. This is the content of Theorem 2.4 in
\cite{huber1995} which we now adapt to our situation.

Recall that, using \Cref{thm:mfpvert}, the computation of the mixed
fiber polytope of convex polytopes $\poly_0,\dots,\poly_k$ reduces to
the case where $\prjb$ induces a bijection from
$\poly_0+\dots+\poly_k$ to $\prjb(\poly_0+\dots+\poly_k)$.  In the
following lemma we denote by $\varepsilon_i\in \ZZ^{k+1}$ the vector with
$0$ in the $i$th entry and $1$ everywhere else.

\begin{lemma}
  \label{cor:mfform2}
  Let $\mathcal{A}:=(\mathcal{A}_0,\dots,\mathcal{A}_k)$ be a point configuration in
  $\ZZ^{n}$. Denote $\poly_i:=\conv{\mathcal{A}_i}$ and assume that
  $\prjb$ induces a bijection from $\poly_0+\dots+\poly_k$ to
  $\prjb(\poly_0+\dots+\poly_k)$. Let $S$ be a mixed subdivision of
  $\prjb(\mathcal{A}):=(\prjb(\mathcal{A}_0),\dots,\prjb(\mathcal{A}_k))$. Then
  \[
    \mmo(\poly_0,\dots,\poly_k)  = \sum_{i=0}^k\sum_{\substack{C\in S\\ \type(C) = \varepsilon_i}} \vol(C)\cdot \prja(\prjb|_{\mathcal{A}_i}^{-1}(C_i)).
  \]
\end{lemma}
\begin{proof}
  For $\lambda\in \RR^{k+1}$, denote
  $\poly_{\lambda}:=\lambda_0\poly_0+\dots + \lambda_k\poly_k$. First note that
  \[\mmo(\poly_{\lambda}) = \sum_{C\in S} \mmo(\lambda_0\conv{C_0}+\dots+\lambda_k\conv{C_k}).\]
  Therefore we assume that $\prjb(\mathcal{A})$ defines a mixed subdivision of
  itself. This is the case if, denoting $d_i:=\dim(\prjb(\poly_i))$,
  we have $\sum_{i=0}^kd_i = k$.  Since $\prjb$ is linear and induces a
  bijection from $\poly_{\lambda}$ to $\prjb(\poly_{\lambda})$ for
  $\lambda_0,\dots,\lambda_k>0$, it also induces a bijection from
  $\poly_i$ to $\prjb(\poly_i)$. Further, if
  $s_i:\prjb(\poly_i)\rightarrow \poly_i$ denotes the corresponding inverse map,
  then $s_i$ is an affine map.

  Additionally, since $\sum_{i=0}^kd_i = k$ we can choose, for each
  $i=0,\dots,k$, a (possibly empty) set of coordinates $x_i$ on
  $\affspan{\prjb(\mathcal{A}_i)}$ such that $x:=(x_0,\dots,x_k)$ gives
  coordinates on $\affspan{\prjb(\mathcal{A})}$. If we now define
  $s(x):=\sum_{i=0}^ks_i(x_i)$ then the Minkowski integral of
  $\poly_{\lambda}$ consists of the singleton
  \begin{align*}
    \int_{\prjb(\poly_{\lambda})}s(x)dx &=\sum_{i=0}^k\int_{\lambda_i\prjb(\poly_i)}s_i(x_i)dx_i\cdot \prod_{j\neq i}\vol(\prjb(\lambda_j\poly_j))\\
    &=\sum_{i=0}^k\lambda_i^{d_i+1}\prod_{j\neq i}\lambda_j^{d_j}\int_{\poly_i}s_i(x_i)dx_i\cdot \prod_{j\neq i}\vol(\prjb(\poly_j))\\
    &=\prod_{i=0}^k\lambda_i^{d_i}\sum_{i=0}^k\lambda_i\int_{\prjb(\poly_i)}s_i(x_i)dx_i\cdot \prod_{j\neq i}\vol(\prjb(\poly_j)).
  \end{align*}
  Hence we see that the coefficient of $\lambda_0\dots \lambda_k$ in this integral
  is nonzero if and only if, for some $i$, we have $d_i=0$ and for all
  $j\neq i$ we have $\lambda_j=1$. In this case
  $\int_{\prjb(\poly_i)}s_i(x_i)dx_i$ is just the singleton
  $\mathcal{A}_i$ and
  $\prod_{j\neq i}\vol(\prjb(\poly_j)) = \vol(\prjb(\poly))$.  This concludes
  the proof.
\end{proof}

To state and prove our main theoretical result, we need one last
notion of subdivisions of point configurations:

\begin{definition}[$\prjb$-Coherent Subdivision]
  Let $\mathcal{A}:=(\mathcal{A}_0, \dots, \mathcal{A}_k)$ be a point configuration in
  $\ZZ^n$. Fix a generic covector $\gamma\in (\RR^{n-k})^{*}$. Define
\[w_{\gamma}\in \RR^{\abs{\prjb(\mathcal{A}_0)}}\times \dots \times \RR^{\abs{\prjb(\mathcal{A}_k)}}\]
via
\[w_{\gamma,a}:=\min\setdes{\gamma(a')}{a'\in \prjb|_{\mathcal{A}_i}^{-1}(a)},\; a\in \mathcal{A}_i.\]
The coherent subdivision $S_{w_{\gamma}}$ is called a {\em
  $\prjb$-coherent subdivision} of $\prjb(\mathcal{A})$.
\end{definition}

Note that the cells of $S_{w_{\gamma}}$ coincide with the images of the
$k$-dimensional faces of $\mathcal{A}$ under $\prjb$ on which $\gamma$ attains its
fiberwise minimum. These are precisely the $k$-dimensional faces
over which we we are summing in \Cref{thm:mfpvert}.

Finally we are ready to state our main theoretical result:

\begin{theorem}
  \label{thm:main}
  Let $\mathcal{A}:=(\mathcal{A}_0, \dots, \mathcal{A}_k)$ be a point configuration in
  $\ZZ^n$. Denote $\poly_i:=\conv{\mathcal{A}_i}$. For a generic covector
  $\gamma\in (\RR^{n-k})^{*}$ let $\mathcal{A}_{\gamma}$ be the point configuration
  consisting of the sets
  \[ \mathcal{A}_{i,\gamma}:=\setdes{\argmin_{\gamma}\prjb|_{A_i}^{-1}(\prjb(a))}{a\in \mathcal{A}_i}. \]
  Let $S$ be a mixed subdivision of $\prjb(\mathcal{A})$ which
  refines the $\prjb$-coherent subdivision of $\prjb(\mathcal{A})$ given by $\gamma$. Then
  \[\mmo(\poly_0,\dots,\poly_k)^{\gamma} = \sum_{i=0}^k\sum_{\substack{C\in S\\ \type(C) = \varepsilon_i}}\vol(C)\cdot \prja(\prjb|_{\mathcal{A}_{i,\gamma}}^{-1}(C_i)).\]
\end{theorem}
\begin{proof}
  According to \Cref{thm:mfpvert}, we have
  \[\mmo(\poly_{\bullet})^{\gamma} = \sum_{\fc_{\bullet}\in \mathcal{F}_{k,\gamma}}
    \mmo(\fc_{\bullet}).\] Define
  $\poly_{i,\gamma}:=\conv{\mathcal{A}_{i,\gamma}}$ and
  $\poly_{\gamma}:=(\poly_{0,\gamma},\dots,\poly_{k,\gamma})$.

  Note now that we have
  \begin{align*}
    \mathcal{F}_{k,\gamma} &= \setdes{\fc\in \mathcal{F}_k(\poly)}{\forall x\in \sum_i\fc_i : x = \argmin_{\gamma}\prjb^{-1}(\prjb(x))}\\
            &= \setdes{\fc\in \mathcal{F}_k(\poly_{\gamma})}{\forall x\in \sum_i\fc_i : x = \argmin_{\gamma}\prjb^{-1}(\prjb(x))}.
  \end{align*}
  To conclude the proof, we apply \Cref{cor:mfform2} to each
  face of $\mathcal{A}_{\gamma}$ which corresponds to a face in
  $\mathcal{F}_{k,\gamma}$, with the mixed subdivision $S$ restricted to the image of
  this face of $\mathcal{A}_{\gamma}$.
\end{proof}

In pseudocode, \Cref{thm:main} yields \Cref{alg:vert} to compute
vertices of mixed fiber polytopes.

\begin{algorithm}
  \caption{Computing Vertices of Mixed Fiber Polytopes}
  \label{alg:vert}
  \raggedright

  \begin{description}
  \item[Input] Finite sets
    $\mathcal{A}_0,\dots,\mathcal{A}_k\subset \ZZ^n$, a covector $\gamma\in (\RR^{n-k})^{*}$.
  \item[Output] The vertex $\mmo(\poly_0,\dots,\poly_k)^{\gamma}$ of the
    mixed fiber polytope $\mmo(\poly_0,\dots,\poly_k)$ where
    $\poly_i:=\conv{\mathcal{A}_i}$.
  \end{description}

  \begin{pseudo}
    \kw{for} $i=0,\dots,k$\\+
    $\mathcal{A}_{i,\gamma}\gets\setdes{\argmin_{\gamma}\prjb|_{A_i}^{-1}(\prjb(a))}{a\in \mathcal{A}_i}$\\
    $w_{i,\gamma}\gets (\min_{\gamma}\prjb|_{\mathcal{A}_i}^{-1}(a)\;|\; a\in \prjb(\mathcal{A}_i))$\\-
    $w_{\gamma}\gets (w_{0,\gamma},\dots,w_{k,\gamma})$\\
    Choose $\epsilon > 0$\label{ln:eps}\\
    Choose $w\in \RR^{\abs{\prjb(\mathcal{A}_0)}}\times \dots \times \RR^{\abs{\prjb(\mathcal{A}_k)}}$ s.t. $\norm{w-w_{\gamma}} < \epsilon$\label{ln:w}\\
    $S_w\gets $ the mixed subdivision of $\prjb(\mathcal{A})$ induced by $w$\\
    $v\gets 0\in \RR^{n-k}$\\
    \kw{for} $i=0,\dots,k$\\+
    \kw{for} $C\in S_w$\\+
    \kw{if} $\type(C)=\varepsilon_i$\\+
    $v\gets v + \vol(C)\cdot \prja(\prjb|_{\mathcal{A}_{i,\gamma}}^{-1}(C_i))\in \RR^{n-k}$\\---
    \kw{return} $v$
  \end{pseudo}
\end{algorithm}

\begin{theorem}
  \label{thm:algcor}
  On input $\mathcal{A}_0,\dots,\mathcal{A}_k\in \ZZ^n$ and
  $\gamma\in (\RR^{n-k})^{*}$ \Cref{alg:vert} is correct under the following
  conditions:
  \begin{enumerate}
  \item \label{it:1} $\gamma$ is sufficiently generic, i.e. lies outside a
    union of finitely many linear subspaces of $(\RR^{n-k})^{*}$.
  \item \label{it:2} $\epsilon$ is chosen sufficiently small in line \ref{ln:eps}. 
  \item \label{it:3} $w$ is chosen sufficiently generic in line \ref{ln:w},
    i.e. lies outside a union of finitely many linear subspaces of
    $\RR^{\abs{\prjb(\mathcal{A}_0)}}\times \dots \times \RR^{\abs{\prjb(\mathcal{A}_k)}}$.
  \end{enumerate}
\end{theorem}
\begin{proof}
  If condition \ref{it:1} is met, then
  $\mmo(\poly_0,\dots,\poly_k)^{\gamma}$ is indeed a vertex of
  $\mmo(\poly_0,\dots,\poly_k)$. If the conditions \ref{it:2} and
  \ref{it:3} are met, then, according to \Cref{lem:refine}, $S_w$ is a
  mixed subdivision of $\prjb(\mathcal{A})$ which refines the
  $\prjb$-coherent subdivision of $\prjb(\mathcal{A})$ induced by
  $\gamma$. The claimed correctness then follows from \Cref{thm:main}.
\end{proof}

Let $\mathcal{A}:=(\mathcal{A}_0,\dots,\mathcal{A}_k)$ be a point configuration in
$\ZZ^n$ and let $\poly_i:=\conv{\mathcal{A}_i}$. While assumptions \ref{it:1}
and \ref{it:3} in \Cref{thm:algcor} are probabilistic assumptions,
true with probability one, it is interesting to ask what can go wrong
if assumption \ref{it:2} is not satisfied. A careful reading of the
proof of \Cref{cor:mfform2} reveals that, in this case,
\Cref{alg:vert} computes an integer point contained in the mixed fiber
polytope $\mmo(\poly_0,\dots,\poly_k)$, but not necessarily a
vertex. Indeed, the section $s$ of $\poly:=\poly_0+\dots+\poly_k$ used
in this proof (and extended over all cells of the mixed subdivision
$S_w$ used by \Cref{alg:vert}) is always a section of $\poly$ even if
$S_w$ does not refine the $\prjb$-coherent subdivision of
$\prjb(\mathcal{A})$ used in \Cref{alg:vert}.

\begin{wrapfigure}{l}{0.6\textwidth}
  \centering
  \begin{tikzpicture}[thick,fill opacity=0.9,scale=0.45,every node/.style={scale=0.8}]
    
    % triangal lines
    \draw[-, black] (0,0) -- (0,7);
    \draw[-, black] (0,0) -- (7,0); 
    \draw[line width=2.5pt][-, black] (7,0) -- (0,7);
    
    % fiber line
    \draw[dashed] (4,-2.5) -- (4,9);

    % OX
    \draw[->, black] (-2,-2) -- (9,-2);
    \draw (9,-2) node[right]{\Large $\mathbb{R}^{k}$};
    
    % OY
    \draw[->, black] (-2,-2) -- (-2,9);
    \draw (-2,9) node[above]{\Large $\mathbb{R}^{n-k}$};

    % gamma
    \draw[-latex, black] (2,4) -- (2,1);
    \draw (2,2.5) node[right]{\Large $\gamma$};
    
    % nodes
    \node[circle, draw=black, fill=col0, inner sep=0pt,minimum size=0.4em] at (4,3) { };
    \draw (4.3,3) node[above]{\large $x$};

    \node[circle, draw=black, fill=col0, inner sep=0pt,minimum size=0.4em] at (4,0) { };
    \draw (4.3,0) node[above]{\large $x'$};

    \node[circle, draw=black, fill=col0, inner sep=0pt,minimum size=0.4em] at (4,-2) { };
    \draw (4.5,-2) node[above]{\large $p(x)$};

    \draw (0,7) node[left]{\large $(\delta, 0)$};    
    \draw (0,0) node[left]{\large $(0, 0)$}; 
    \draw (7,0) node[right]{\large $(0, \delta)$};     
  \end{tikzpicture}
  \caption{For $\delta\in \set{d_0,d_1}$, the covector $\gamma:=-1 \in \RR^{*}$ attains its fiberwise minimum
    on the convex hull of $(\delta,0)$ and $(0,\delta)$.}
  \label{fig:example}
\end{wrapfigure}
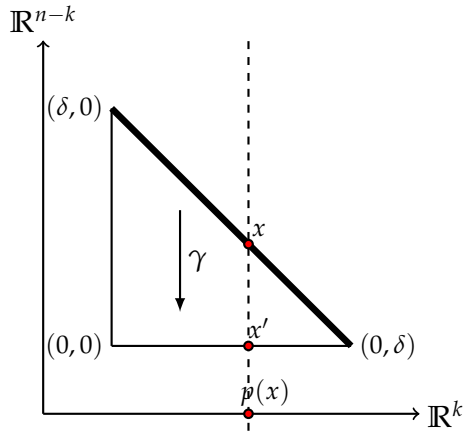

Consequently, regardless of the choice of $\epsilon$ in \Cref{alg:vert},
the algorithm in \cite{huggins2006} which we use to construct the
mixed fiber polytope $\mmo(\poly_0,\dots,\poly_k)$ always terminates
and computes a polytope $\poly'$ contained in
$\mmo(\poly_0,\dots,\poly_k)$.  When we are using this algorithm to
solve an implicitization problem, as sketched in the introduction, the
correctness of $\poly'$ may then a posteriori be verified by
attempting to interpolate the desired eliminant from $\poly'$, as in
\cite{rose2025}.  If this interpolation attempt fails, we can then
restart the computation of $\mmo(\poly_0,\dots,\poly_k)$ with a
smaller choice of $\epsilon$. \Cref{lem:refine} garantuees that this will
succeed after finitely many iterations.

We now illustrate \Cref{thm:main} and \Cref{alg:vert} with the following
example.

\begin{example}
  Let $n=2$, $k = 1$ and $d_0,d_1\in \ZZ_{>0}$. Let
  \[\mathcal{A}_0 = \set{\twov{0}{0}, \twov{0}{d_0}, \twov{d_0}{0}},\quad
    \mathcal{A}_1 = \set{\twov{0}{0}, \twov{0}{d_1}, \twov{d_1}{0}}.\] Let
  $\poly_0:=\conv{\mathcal{A}_0}$ and $\poly_1:=\conv{\mathcal{A}_1}$.  The polytopes
  $\poly_0$ and $\poly_1$ are triangles. If $f_0,f_1\in \CC[x,y]$ are
  generic of degree $d_0$ and $d_1$ respectively, then their Newton
  polytopes are exactly $\poly_0$ and $\poly_1$ and we expect their
  composite polynomial to be of degree $d_0d_1$, matching the Bezout
  bound associated to $f_0$ and $f_1$.  Correspondingly, in light of
  \Cref{thm:mfpelim}, we expect $\mmo(\poly_0,\poly_1)$ to be the
  interval $[0,d_0d_1]\subset \RR$. Note that at the vertex $d_0d_1$ of this
  interval the covector in $\gamma\in(\RR^1)^{*}$ given by multiplication with
  $-1$ is minimized.

  Let us therefore now compute the vertex
  $\mmo(\poly_0,\poly_1)^{\gamma}$ using \Cref{thm:main}/\Cref{alg:vert}. 

  In the notation of \Cref{thm:main}, the point configuration $\mathcal{A}_{\gamma}$ is given by 
  \[\mathcal{A}_{\gamma,0} = \set{\twov{0}{d_0}, \twov{d_0}{0}},\quad\mathcal{A}_{\gamma,0} =
    \set{\twov{0}{d_1}, \twov{d_1}{0}},\]
  see also \Cref{fig:example}.

  The configuration $\mathcal{A}_{\gamma}$ is itself a $1$-dimensional face of
  $\mathcal{A}$. Hence the $\prjb$-coherent subdivision of
  $\prjb(\mathcal{A})$ induced by $\gamma$ is the trivial one, consisting of
  $\prjb(\mathcal{A})$ itself. It is therefore in particular refined by any
  mixed subdivision.

  We have
  \[\prjb(\mathcal{A}_0) = \set{0,d_0},\quad
    \prjb(\mathcal{A}_1) = \set{0,d_1}.\]
  A mixed subdivision of $\prjb(\mathcal{A})$ is given by the two cells
  \begin{align*}
    &C_0 := (C_{0,0},C_{0,1}) := (\set{0},\set{0,d_1}),\\
    &C_1 = (C_{1,0}, C_{1,1}) := (\set{0,d_0},\set{d_1}),
  \end{align*}
  see also \Cref{ex:mixsub}. We further compute
  \begin{align*}
    &\prjb|_{\mathcal{A}_{\gamma,0}}^{-1}(C_{0,0}) = \set{\twov{d_0}{0}};\;\vol(C_{0}) = d_1,\\
    &\prjb|_{\mathcal{A}_{\gamma,1}}^{-1}(C_{1,1}) = \set{\twov{0}{d_1}};\;\vol(C_{1}) = d_0.\\
  \end{align*}
  Plugging this into \Cref{thm:main} yields
  \begin{align*}
    \mmo(\poly_{0},\poly_1)^{\gamma} &= d_1\cdot \prja\left(\set{\twov{d_0}{0}}\right) + d_0 \cdot  \prja\left(\set{\twov{0}{d_1}}\right)\\
    &=d_1\cdot\set{d_0} + d_0\cdot \set{0} = \set{d_0d_1}
  \end{align*}
  as expected.
\end{example}

\section{Implementation \& Experiments}
\label{sec:exam}

The source code of our implementation accompanying this paper
can be found at
\begin{center}
  \url{https://github.com/RafaelDavidMohr/MixedFiberPolytope}.
\end{center}
This implementation is written in the programming language
\texttt{julia} \cite{bezanson2017}. It includes \Cref{alg:vert} as
well as an implementation of the algorithm given in
\cite{huggins2006}. This algorithm computes a polytope $\poly$ given a
{\em vertex oracle} for $\poly$, i.e. a black box that returns the
vertex $\poly^{\gamma}$ for a generic covector $\gamma$. \Cref{alg:vert}
provides such a black box for mixed fiber polytopes.

To compute mixed subdivisions we rely on the \texttt{julia}-package
\texttt{MixedSubdivisions.jl} which itself is based on
\cite{jensen2016}. For necessary operations with polytopes, we make
use of the computer algebra system \texttt{Oscar} \cite{OSCAR-book}
which itself uses \texttt{polymake} \cite{assarf2017} to perform these
computations.

\subsection*{List of Examples}
\label{sec:exmpllist}

We compared our implementation to \texttt{TropicalImplicitization.jl}
\cite{rose2025} on the following implicitization problems:

\begin{enumerate}
\item \label{ex:1} In $\CC[x_0,x_1,x_2,x_3,s,t,u]$ we let
  \begin{align*}
    &f_0 = x_0- (s^2t + s^2u + 4stu + 3su^2 + 2t^3 + 4t^2u + 2tu^2 + 2u^3),\\
    &f_1 = x_1- ( -s^3 - 2s^2u - 2st^2 - stu + su^2 - 2tu^2 + 2u^3),\\
    &f_2 = x_2 - (-s^3 - 2s^2t - 3s^2u - 3st^2 - 3stu - 2su^2 + 2t^2u - 2tu^2),\\
    &f_3 = x_3 - (s^3 + s^2t + s^2u - su^2 + t^3 + t^2u - tu^2 - u^3).
  \end{align*}
  This is Example 3.3.4 in \cite{buse2005}.
% \item \label{ex:2} In $\CC[x_0,x_1,x_2,x_3,x,y,z]$ we let
%   \begin{align*}
%     &f_0:=x_0-(xy+z+1),f_1:=x_1-(xz+y+1),\\
%     &f_2:=x_2-(yz+x+1),f_3:=x_3-(x^3+y^5+z^7).
%   \end{align*}
%   This example appears on page 113 of \cite{sturmfels2008}.
\item \label{ex:3} In $\CC[x_0,x_1,x_2,x_3,x_4,x_5,s,t,u,v,w]$ we let
  \begin{align*}
    &f_0:=x_0-(s^3 - u^2 - t - 3s - u +w),\\
    &f_1:=x_1-(u^2-sw-11),\;f_2:=x_2-(s^2-5u-v),\\
    &f_3:=x_3-(u^2-s-v-w),\;f_4:=x_4-(u^2+7s+t),\\
    &f_5:=x_5-(v^2+s^2-s-t-w).
  \end{align*}
  This is Example A.22 in \cite{abbott2017}.
\item Another way to generate examples is to take any square
  polynomial system $f_0,\dots,f_k$ in variables $y_0,\dots,y_k$ and
  consider the implicitization problem given by the equations
  \[x_i-f_i(a,y_1,\dots,y_k);\; i = 0,\dots,k,\] where $a$ is a
  randomly chosen value. Using this method, we generated examples for the
  well-known benchmark systems
  \href{https://gitlab.lip6.fr/eder/msolve-examples/-/raw/master/zero-dimensional/kat6-qq.ms}{katsura
    6}, \href{https://mathexp.eu/mohr/cyclic5.ms}{cyclic 5},
  \href{https://gitlab.lip6.fr/eder/msolve-examples/-/raw/master/zero-dimensional/noon6-qq.ms}{noon
    6} and
  \href{https://gitlab.lip6.fr/eder/msolve-examples/-/raw/master/zero-dimensional/noon7-qq.ms}{noon
    7}.
\end{enumerate}

For each of these examples we extracted the supports
$\mathcal{A}_0,\dots,\mathcal{A}_k$ of the defining equations. The resulting point
configuration is the input of our algorithm. The running times (in
seconds, unless otherwise stated) for these examples are recorded in
\Cref{tab:bench}. We rounded any time below 0.1 seconds to zero. In
addition, we record in \Cref{tab:bench} the number of mixed
subdivisions computed by our algorithm and the number of faces of
dimension $k+1$ of the Minkowski sum $\poly_0+\dots + \poly_k$. As
mentioned in the introduction, for each of these faces tropical
methods have to compute one mixed subdivision.

\Cref{tab:bench} illustrates that our algorithm computes vastly fewer
mixed subdivisions than tropical elimination on the chosen examples.
For katsura 6 and noon 6, just the computation of the faces of dimension
$k+1$ of the associated Minkowski sum exceeded the time indicated
in the 4th column. For noon 7, this computation ran out of memory.

The following software versions were used to record the data in
\Cref{tab:bench}: \texttt{julia} @ v1.11.3, \texttt{Oscar} @ v1.2.2,
\texttt{MixedSubdivisions.jl} @ v1.1.5 and
\texttt{TropicalImplicitization.jl} @ v1.0.0. The computations were
performed on a laptop with 32GB of memory and an Intel Core i7-1185G7
CPU.

\begin{table}[H]
\centering
\captionsetup{justification=centering}
\caption{Running time (in sec) for our examples}
\begin{tabular}{|l|c|c|c|c|}
  \hline
  \multirow{2}{*}{Example} & \multicolumn{2}{c|}{Our Algorithm} & \multicolumn{2}{c|}{Tropical Elimination} \\ \cline{2-5}
                    & Time & \$ & Time & \pounds \\
  \hline \hline
  Example \ref{ex:1} & 0 & 20 & 5.57 & 83\\
  \hline
  Example \ref{ex:intro} & 0.1 & 38 & 39.7 & 651\\
  \hline
  Example \ref{ex:3} & 0.1 & 75 & >15 min. & 72214\\
  \hline
  katsura 6 & 0.1 & 71 & >15 min. & ? \\
  \hline
  cyclic 5 & 0.5 & 160 & >15 min. & 13469 \\
  \hline
  noon 6 & 10.1 & 1722 & >30 min. & ? \\
  \hline
  noon 7 & 345.6 & 18600 & OOM & ?\\
  \hline
\end{tabular}
\caption*{\$: Number of mixed subdivisions computed

  \pounds: Number of faces of dim. $k+1$ of Minkowski sum

  OOM: Out of memory}
\label{tab:bench}
\end{table}

\section{An Application to Differential Elimination}
\label{sec:diffelim}

An interesting potential application of our algorithm comes
from {\em differential elimination}, which we illustrate here.

We consider a {\em dynamical system} of the shape
\begin{equation}
  \label{eq:dyn}
 x'_1 - g_1 = \cdots = x'_n - g_n = 0,
\end{equation}
where $g_i\in \CC[x_1,\dots,x_n]$. The $x_i$ should be thought of as
differentiable functions $x_i(t)$ in a time variable
$t$. Frequently, one wishes to compute a {\em minimal} differential
relation $\minpol$ satisfied by the coordinate $x_1$ of a solution
vector $(x_1,\dots,x_n)$ to \eqref{eq:dyn}. This problem was recently
considered by the second author of this work in \cite{mukhina2025}. It
turns out that $\minpol\in \CC[x_1,x_1',\dots,x_1^{(n)}]$. In
\cite{mukhina2025}, a polytope $\poly_{\text{min}}$ containing the
Newton polytope of $\minpol$ was constructed by giving explicit
inequalities satisfied by the elements in the support of $\minpol$
(Theorem 1 in \cite{mukhina2025}). Then one can compute $\minpol$ out
of $\poly_{\text{min}}$ (Algorithm 1 in \cite{mukhina2025}) using an
evaluation-interpolation approach by solving a linear system with number of
equations and variables equal to the number of lattice points in
$\poly_{\text{min}}$.

We consider an example. Write $d:=\deg g_1$,
$D:= \max_{2\leq i \leq n}\deg g_i$ and consider the system
\begin{equation*}
  \begin{cases}
    x'_1 = x_1^2 + x_1 x_2 + x_2^2 + 1,\\ 
    x_2' =x_2.
  \end{cases}
\end{equation*}
The bounds obtained in \cite{mukhina2025} depend on $n,d$ and $D$. In
this case $n = 2, d = 2, D = 1$.  We introduce a new variable $y$ and write
$y, x_1', x_1''$ as polynomials in the variables $x_1, x_2, x_1'$:
\begin{equation}
  \label{eq:exdyn}
  \begin{split}
    &y = x_1,\\
    &x_1' = x_1^2 + x_1 x_2 + x_2^2 + 1,\\
    &x_1'' = x'_1 (2 x_1 + x_2) + x_2(2 x_2 + 1).
  \end{split}
\end{equation}
From this system, one can then compute $\minpol$ by eliminating the
variables $x_1,x_2$.

Note that to the system \eqref{eq:exdyn} we can also apply our
algorithm by computing the mixed fiber polytope w.r.t. the variables
$y,x_1,x_1'$. This mixed fiber polytope will contain the Newton
polytope of $\minpol$ and $\minpol$ can then be computed using this
mixed fiber polytope in combination with Algorithm 1 in
\cite{mukhina2025}. It is then interesting to ask how many lattice
points this mixed fiber polytope has compared to $\poly_{\text{min}}$,
both from a theoretical viewpoint (i.e. is it possible to improve the
bounds depending on $n,d$ and $D$ in \cite{mukhina2025} using mixed
fiber polytopes) and from a practical viewpoint (i.e. is it beneficial
to use mixed fiber polytopes instead of $\poly_{\text{min}}$ for
specific dynamical systems to exploit their sparsity patterns).

For the case $d\leq D$, the obtained bounds were shown to be sharp in
\cite{mukhina2025} but obtaining sharp bounds for the case $d>D$
remains open.

In \Cref{tab:mixedfiber} we compare the number of lattice points in
$\poly_{\text{min}}$ with the number of lattice points in the
associated mixed fiber polytope. The column $\%$ shows the difference
between these numbers as a percentage of the first number. The first 6
rows correspond to dynamical systems with random dense polynomials
$g_i$ satisfying the indicated degree constraints. For the case
$n=4,d=3,D=2$, we were unable to compute all integer points in the
polytope extracted out of \cite{mukhina2025} within twelve hours of
computation. The last three rows of \Cref{tab:mixedfiber} correspond
to specific choices of sparse $g_i$.

\begin{table}[H]
  \small
\centering
\caption{Comparison between the bound 
  from \cite{mukhina2025} and mixed fiber polytope computations}
\begin{tabular}{ |c|c|c|c|c| } 
	\hline
	\multirow{2}{*}{$[n,d, D]$} & \multicolumn{2}{c|}{\# of lattice points} & \multirow{2}{*}{\%} \\ \cline{2-3}
    & From \cite{mukhina2025} & Mixed Fiber Polytope & \\
	\hline\hline	
    [3,2,1] & 271 & 266 &  2\%\\ 
	\hline
	[3,3,1] & 9520 & 8661 &  10\%\\ 
	\hline
    [3,3,2] & 25788 & 25525 &  1\%\\ 
	\hline
	\hline
	[4,2,1] & 11021 & 10798  &  2\%\\ 
	\hline
	[4,3,1] & 10639882 & 9694938  &  9\% \\ 
    \hline
    [4,3,2] & ? & 391478920 &   \% \\
  \hline
  \hline
  Lotka Voltera 3D \cite{mukhina2025} & 31757 & 21341 & 32\%\\
  \hline
  Lotka Voltera 4D \cite{wangersky1978} & 2379301 & 15086 & 99\%\\
  \hline
  Blue Sky Catastrophe \cite{vangorder2013} & 65637 & 26698 & 59\% \\
  \hline
\end{tabular}
\label{tab:mixedfiber}
\end{table}

\Cref{tab:mixedfiber} indicates that, while mixed fiber polytopes
probably do not yield a substantial improvement of Theorem 1 in
\cite{mukhina2025}, our algorithm can improve practical computations
for specific sparse systems by shrinking the size of the linear system
to solve to compute $\minpol$.

\newpage
\printbibliography
\end{document}